\documentclass{article}
\usepackage[utf8]{inputenc}
\usepackage{amsmath, amssymb, mathtools, graphicx, amsthm, dsfont, float, bbm, hyperref, caption, adjustbox}
\usepackage{xcolor}
\usepackage{comment}
\usepackage{cleveref}
\usepackage{booktabs}
\usepackage[round]{natbib}

\renewcommand{\le}{\leqslant}
\renewcommand{\ge}{\geqslant}
\renewcommand{\leq}{\leqslant}
\renewcommand{\geq}{\geqslant}
\newcommand{\phz}{\phantom{0}}

\newcommand{\giv}{\!\mid\!}

\newcommand{\rd}{\,\mathrm{d}}

\newtheorem{theorem}{Theorem}

\newtheorem{lemma}{Lemma}
\theoremstyle{definition}
\newtheorem{definition}{Definition}

\newcommand{\tr}{\mathrm{tr}}
\newcommand{\cov}{\mathrm{Cov}}
\newcommand{\var}{\mathrm{Var}}
\newcommand{\dunif}{\mathcal{U}}
\newcommand{\dnorm}{\mathcal{N}}

\newcommand{\bbR}{\mathbb{R}}
\newcommand{\bbP}{\mathbb{P}}
\newcommand{\bbE}{\mathbb{E}}
\newcommand{\Var}{\text{Var}}

\newcommand{\st}{\text{ } | \text{ }}
\newcommand{\eps}{\varepsilon}

\newcommand{\indic}{\mathds{1}}

\newcommand{\simiid}{\stackrel{\mathrm{iid}}\sim}
\newcommand{\mycomment}[1]{}

\title{Multivariate Tie-breaker Designs}
\author{Tim P. Morrison\\Stanford University \and Art B. Owen\\Stanford University }
\date{October 2024}

\begin{document}

\maketitle

\begin{abstract}
In a tie-breaker design (TBD), subjects with high values of a running variable are given some (usually desirable) treatment, subjects with low values are not, and subjects in the middle are randomized. TBDs are intermediate between regression discontinuity designs (RDDs) and randomized controlled trials (RCTs).
    TBDs allow a tradeoff between the resource allocation efficiency of an RDD and the statistical efficiency of an RCT. We study a model where the expected response is one multivariate regression for treated subjects and another for control subjects. We propose a prospective D-optimality, analogous to Bayesian optimal design, to understand design tradeoffs without reference to a specific data set. For given covariates, we show how to use convex optimization to choose treatment probabilities that optimize this criterion. We can incorporate a variety of constraints motivated by economic and ethical considerations. In our model, D-optimality for the treatment effect coincides with D-optimality for the whole regression, and, without constraints, an RCT is globally optimal. 
    We show that a monotonicity constraint favoring more deserving subjects induces sparsity in the number of distinct treatment probabilities. We apply the convex optimization solution to a semi-synthetic example involving triage data from the MIMIC-IV-ED database.
\end{abstract}

\section{Introduction}

There are many important settings where a costly or scarce
intervention can be made for some but not all subjects.
Examples include giving a scholarship to some students \citep{Angrist2014leveling}, 
intervening to prevent child abuse in some
families \citep{kran:2022}, programs to counter
juvenile delinquency \citep{lips:cord:berg:1981}
and sending some university students to remedial
English classes \citep{aike:west:schw:carr:hsiu:1998}. 
There are also lower-stakes settings where a company might
be able to offer a perk such as a free service upgrade,
to some but not all of its customers.
In settings such as these, there is usually a priority ordering
of subjects, perhaps based on how deserving they are or on how much
they or the investigator might gain from the intervention. 
We can represent that priority order in terms of a real-valued 
running variable $x_i$ for subject $i$.

To maximize the immediate short-term value from the limited
intervention, one can assign it only to subjects with $x_i\ge t$
for some threshold $t$.  The difficulty with this ``greedy'' solution
is that such an allocation makes it difficult to estimate the causal
effect of the treatment. It is possible to use a regression 
discontinuity design (RDD) in this case, comparing subjects with
$x_i$ somewhat larger than $t$ to those with $x_i$ somewhat smaller
than $t$. For details on the RDD, see the comprehensive recent 
survey by \cite{cattaneo2022regression}. The RDD can give a consistent
nonparametric estimate of the treatment effect at $x=t$ \citep{hahntodd}, but
not at other values of $x$. The treatment effect can be studied
at other values of $x$ under an assumed regression model, but then
the variance of the regression coefficients can be large due
to the strong dependence between the running variable and 
the treatment \citep{gelman2019high}.

The RDD is commonly used to analyze observational data for which the  investigator has no control over the treatment cutoff. In the settings we consider, the investigators assign the treatment and can therefore employ some randomization.  The motivating context makes it costly or even ethically difficult to use a randomized controlled trial (RCT) where the treatment is assigned completely at random without regard to $x_i$. The tie-breaker design (TBD) is a compromise between the RCT and RDD.  It is a triage where subjects with large values of $x_i$ get the treatment, those with small values get the control level and those in between are randomized to either treatment or control.

While the tie-breaker design has been known since \cite{camp:1969} it has not been subjected to much analysis.  Most of the theory for the TBD has been for the setting with a scalar $x$ and models that are linear in $x$ and the treatment \citep{owenvarian, li:owen:2022:tr}.  In this paper we study the TBD for a vector-valued predictor. Our first motivation is that many use cases for TBDs will include multiple covariates. Second, although multivariate nonparametric regression models are out of the scope of this paper, we believe that TBD regressions are a useful first step in that direction.  Third,
\cite{gelm:hill:veht:2020} counsel against RDDs
 that do not adjust for pre-treatment variables. The multivariate regression setting supports such adjustments. 

In a tie-breaker design we have a vector of covariates for subject $i$ and
must choose a two-level treatment variable. We then face an atypical experimental design problem where some of the predictors are fixed and not settable by us while one of them is subject to randomization.  This is known as the `marginally restricted design' problem after \cite{cookthibodeau} who studied $D$-optimality in such a setting.  We consider two such settings.  In one setting, the investigator already has the covariate values. The second setting is one step removed, where we consider random covariates that an investigator might later have.

The paper is organized as follows. Section~\ref{sec:setup} outlines our notation and the regression model.
Section~\ref{sec:efficiency} introduces our notions of efficiency and D-optimality in the multivariate regression model.  Theorem~\ref{thm:dopt=dopt} shows that
D-optimality for the treatment effect parameters is equivalent to D-optimality for the full regression.
To study efficiency for future subjects, we use a prospective D-optimality criterion, adapted from Bayesian optimal design, that maximizes expected future information.
Theorem~\ref{thm:rctisoptimal} then shows that the RCT is prospectively D-optimal. We also discuss an example involving Gaussian covariates and a symmetric design, which provides relevant intuition. 
Section~\ref{sec:shorttermgain} finds an expression for the expected short-term gain, with particular focus on this Gaussian case. When the running variable is linear in the covariates, the best linear combination for statistical efficiency is the ``last'' eigenvector of the covariance matrix of covariates while the best linear combination for short-term gain is the true treatment effect. Section~\ref{sec:convex} presents a design strategy based on convex optimization to choose treatment probabilities for given covariates and compares the effects of various practical constraints. In particular, we show that a monotonicity constraint in the treatment probabilities yields solutions with a few distinct treatment probability levels. This is consistent with some past results in constrained experimental design \citep{cookfedorov}, but both our proof and the theorem particulars are distinct. Section~\ref{sec:mimic} illustrates this procedure on a hospital data set from MIMIC-IV-ED about which emergency room patients should receive intensive care.
Section~\ref{sec:discussion} has a brief discussion of some additional context for our results.

\subsection{Tie-breaker designs}

The earliest tie-breaker design of which we are aware  \citep{camp:1969} had a discrete running variable, and literally broke a tie by randomizing the treatment assignments for those subjects with $x=t$. For an imperfectly measured running variable, we might consider $|x-t|<\Delta$ to not be meaningfully large for some $\Delta>0$, so randomizing in that window is like breaking ties.  This is the approach from \cite{boru:1975}.

When we randomize for $x$ with $|x-t|< \Delta$, then setting $\Delta=0$ provides an RDD while $\Delta\to\infty$ provides an RCT.  The TBD transitions smoothly between these extremes.  Many of the early comparisons among these methods use the two-line regression model
\begin{align}\label{eq:twoline}
Y_i = \beta_0 + \beta_1X_i + \gamma_0Z_i +\gamma_1Z_iX_i+\eps_i.
\end{align}
Here, $Y_i$ is the response for subject $i$, $Z_i \in \{-1, 1\}$ is a binary treatment variable, and $\eps_i$ is an IID error term. \cite{goldberger} finds  for $X_i\sim\dnorm(0,1)$ that the RDD has about 2.75 times the variance of an RCT.  \cite{jacob2012practical} consider polynomials in $X$ up to third degree, with or without interactions with $Z$ for Gaussian and for uniform random variables. Neither paper considers tie-breakers. \cite{cappelleri1994power} study the two-line model above (absent the $X$-$Z$ interaction) and compare the RDD, RCT and three intermediate TBD designs. They study the sample size required to reject $\gamma=0$ with sufficient power for those designs and three different sizes of $\gamma$. Larger $\Delta$ brings greater power.

\cite{owenvarian} find an expression for the variance of the parameters in the two-line model for $x$ with a symmetric uniform distribution about a threshold $t=0$, and also for a Gaussian case.  The randomization window is $|x|\le\Delta$, with a larger $\Delta$ yielding smaller variance but a lesser short-term gain.  Their designs have only three levels (0, 50 and 100 percent) for the treatment probability given $x$.  They show that there is no advantage to using other levels or a sliding scale when $x$ has a symmetric distribution and $50$\% of the subjects will get the treatment.

\cite{li:owen:2023} study the two-line regression for a general distribution of $x$ and an arbitrary proportion of treated subjects.  They constrain the global treatment probability and a measure of short-term gain. They find that there exists a D-optimal design under these constraints with a treatment probability that is constant within at most four intervals of $x$ values. Moreover, with the addition of a monotonicity constraint, there exists an optimal solution with just two levels corresponding to $x < t'$ and $x > t'$ for some $t'$. 

\cite{klugerowen} consider tie-breaker design for nonparametric regression for real-valued treatments $x$.  A tie-breaker design can support causal inference about the treatment effect at points within the randomization window, not just at the threshold $x=t$.  At the threshold, the widely use kernel regression methods for RDD become much more efficient if one has sampled from a tie-breaker, because they can use interpolation to $x=t$ instead of the extrapolation from the left and right sides that an RDD requires.

\subsection{Marginally constrained experimental design}
The tie-breaker setup is a special case of marginally constrained experimentation of \cite{cookthibodeau}.  Some more general constraints are considered in \cite{lopez2004optimal}. Some of our findings are special cases of more general results in \cite{nachtsheim1989design}. \cite{heavlin1998columnwise} consider a sequential design setting in which the columns of a design matrix correspond to steps in semi-conductor fabrication, with each step constrained by the prior ones.  The marginally constrained literature generally considers choosing $n_i$ values of the settable variables for level $i$ of the fixed variables.  In our setting, we cannot assume that $n_i>1$.

The TBD setting often has a monotonicity constraint on treatment probabilities that we have not seen in the marginally constrained design literature. Under that constraint a more ``deserving'' subject should never have a lower treatment probability than a less deserving subject has.  This and other constraints in the TBD will often lead to optimal designs with a small number of different treatment levels. 

\subsection{Sequential experimentation}

We anticipate that the tie-breaker design could be used in sequential experimentation.  In our motivating applications, the response is measured long enough after the treatment (e.g., six years in some educational settings) that bandit methods \citep{slivkins2019introduction} are not appropriate.  There is related work on adaptive experimental design.  See for instance \cite{metelkina}.  The problem we focus on is designing the first experiment that one might use, and a full sequential analysis is not in the scope of this paper.

\section{Setup}\label{sec:setup}

Given $n$ subjects with $d$ covariates each, we let $X\in\bbR^{n\times d}$ be the design matrix and $X_i \in \bbR^d$ be the variables for subject $i$. We write $\tilde{X} = \begin{bmatrix} 1 & X \end{bmatrix} \in \bbR^{n \times (d + 1)}$ to denote the matrix with an intercept out front, and $\tilde{X} = [1 \quad X^\top]^\top 
\in\bbR^{d+1}$ to denote its $i$'th row. For ease of notation, we zero-index $\tilde{X}$ so that $\tilde{X}_{i0} = 1$ and $\tilde{X}_{ij} = X_{ij}$, for $j=1,\dots,d$.

  We are interested in the effect of some treatment $Z_i \in \{-1, 1\}$ on a future response $Y_i\in\bbR$ for subject $i$, where $Z_i = 1$ corresponds to treatment and $Z_i = -1$ to control. The design problem is to choose probabilities $p_i\in[0,1]$ and then take $\bbP(Z_i=1)=p_i$.
  This differs from the common experimental design framework in which the covariates $X_i$ can also be chosen. 
  In Section~\ref{sec:convex} we will show how to get optimal $p_i$ by convex optimization.

  To get more general insights into the design problem, in Section~\ref{sec:efficiency} we also consider a random data framework.  The predictors are to be sampled with $X_i\simiid P_X$. This allows us to relate
  design findings to the properties of $P_X$ rather than to a specific matrix $X$. After $X_i$ are observed, $Z_i$ will be set randomly and then $Y_i$ observed. The analyst is unable to alter $P_X$ at all but can choose any function $p(X) = \bbP(Z = 1 \giv X)$ that satisfies the imposed constraints.
 
We work with the following linear model:
\begin{align} \label{eq:linmod}
Y_i = \tilde{X}_i^{\top} \tilde{\beta} + Z_i \tilde{X}_i^{\top} \tilde{\gamma} + \eps_i
\end{align}
for $\tilde{\beta}, \tilde{\gamma} \in \bbR^{d + 1}$, where $\eps_i$ are IID noise terms with mean zero and variance $\sigma^2>0$. We use the same notational convention of writing $\tilde\beta=\begin{bmatrix}\beta_0 &\beta^{\top}\end{bmatrix}^{\top}$  and $\tilde{\gamma} = \begin{bmatrix}\gamma_0 \quad \gamma^{\top}\end{bmatrix}^{\top}$ for $\beta,\gamma \in \bbR^d$ to separate out the intercept term.
We consider $\tilde{\gamma}$ to be the parameter of greatest interest because it captures the treatment effect of~$Z$.

Equation~\eqref{eq:linmod} generalizes the two-line model~\eqref{eq:twoline}
studied by \cite{li:owen:2023} and \cite{owenvarian}.
The latter authors describe some computational methods for the model~\eqref{eq:linmod} but most of their theory is for model~\eqref{eq:twoline}.

Though a strong parametric assumption, this multi-line regression is a helpful and practical model to inform treatment assignment at the design stage. Because our scenario of interest could include regions of the covariate space with $p(x) = 0$, we do not have overlap and cannot rely on classical semiparametric methods for causal inference. Nonparametric generalizations of \eqref{eq:linmod} such as spline models may prove more flexible in highly nonlinear settings, but we do not explore them here. 

In Section~\ref{sec:efficiency}, we also consider multivariate tie-breaker designs (TBDs), which we define as follows: the treatment $Z_i$ is assigned via 
\begin{align} \label{eq:probs}
\mathbb{P}(Z_i = 1 \giv X_i) =
\begin{cases}
1, & X_i^{\top}\eta \ge u \\
p, & X_i^{\top}\eta \in (\ell, u) \\
0, & X_i^{\top} \eta \le \ell
\end{cases}
\end{align}
for parameters $u<\ell$, $p\in(0,1)$ and $\eta \in \bbR^d$. That is, we assign treatment to subject $i$ whenever $X_i^{\top}\eta$ is at or above some upper cutoff $u$, do not assign treatment whenever it is at or below some lower cutoff $\ell$, and randomize at some fixed probability $p$ in the middle. For the case $u = \ell$, we take $\bbP(Z_i=1 \giv X_i)=\indic\{X_i^{\top}\eta\ge u\}$ which has a mild asymmetry in offering the treatment to those subjects, if any, that have $X_i^{\top}\eta=\ell=u$.

In equation~\eqref{eq:probs} we ignore the intercept and just consider $X_i$ instead of $\tilde{X}_i$ since the intercept would merely shift everything by a constant. Here, we use $\eta \in \bbR^d$ instead of $\gamma$ from \eqref{eq:linmod} to reflect that the vector we treat on need not be the same as the true $\gamma$, which is unknown.

In practice, $\eta$ will encode whatever constraints on randomization a practitioner may encounter. While linear constraints do not encompass all possibilities, they are simple and flexible enough to account for a great many that one could feasibly seek to impose; for example, constraints on a heart rate below some cutoff or a total SAT score that is sufficiently high can both be covered by \eqref{eq:probs}. One could also set $\ell$ and $u$ to be known or estimated quantiles of a covariate or a linear combination of them if, for example, one wants deterministic treatment assignment for the top $5\%$ of candidates. More generally, we show in Section \ref{sec:convex} how to accommodate any convex constraints on treatment assignment. 

The assignment~\eqref{eq:probs} greatly generalizes the one in \cite{owenvarian} which had $d=1$, $\ell = -u$, $p = 1/2$, and $P_X$ either $\dunif(-1,1)$ or $\dnorm(0,1)$. In analogy to the one-dimensional case, we refer to the case with $u = \ell$ as an RDD. We refer to any choice of $(\ell, u, \eta)$ for which $\mathbb{P}(X^{\top}\eta \in (\ell, u)) = 1$
as an RCT. 

\section{Efficiency and D-optimality}\label{sec:efficiency}

We assume that the covariates $X_i \in \bbR^d$ have yet to be observed and are drawn from some distribution $P_X$ with a finite and invertible covariance matrix $\Sigma$. The aim is to devise a treatment assignment scheme $p(X)$ with $Z_i$ assigned independently via $\bbP(Z_i = 1 \st X_i) = p(X_i)$ that is optimal in some sense. Random allocation of $Z$ has advantages of fairness and supports some randomization-based inference. 

Section \ref{sssec:fixedx} briefly outlines efficiency and D-optimality in the fixed $Z$ setting to introduce relevant formulas. In Section \ref{sssec:prospective}, we then outline a notion of prospective $D$-optimality used when $X$ is random and we must choose a distribution of $Z$ given $X$. We model our approach on the Bayesian optimal design literature \citep{chaloner1995bayesian}, in which the model parameters are treated as random, which we describe in more detail in that section. 

The treatment of $X$ as random permits a theoretical discussion about what is expected to happen depending on the distributional properties of the unseen $X$ data. However, the case in which the covariates are actually fixed is also covered by the procedure in Section \ref{sssec:prospective}, since it corresponds to a point mass distribution on $X$. We return to this case in Section \ref{sec:convex}, in which we show how to use convex optimization to obtain prospectively D-optimal $p(X_i)$. 

\subsection{D-optimality for nonrandom \texorpdfstring{$X$}{X} and \texorpdfstring{$Z$}{Z}}\label{sssec:fixedx}

In this brief section, we outline our notions of efficiency and D-optimality in the setting in which $X$ is fixed and $Z$ are assigned non-randomly. While not the focus of this paper, it will allow us to motivate various definitions and derive a property of D-optimality in the multivariate model that will be useful going forward. 

We begin by conditioning on $X$ and $Z$. Let $D \in \bbR^{n \times n}$ be the 
diagonal matrix whose diagonal entries are $D_{ii} = Z_i$. We can write the linear model \eqref{eq:linmod} in matrix form as $Y = U \delta + \eps$, where $U = \begin{bmatrix} \tilde{X} & D\tilde{X} \end{bmatrix}$  and $\delta = \begin{bmatrix} \tilde{\beta}^{\top} &\tilde{\gamma}^{\top} \end{bmatrix}^{\top}$. 

In the general model~\eqref{eq:linmod},
conditionally on $X$ and $Z$ we have
\begin{align*} \var(\hat\delta \giv X, Z) &=
\begin{bmatrix}
\var(\hat\beta \giv X, Z) & \cov(\hat\beta,\hat\gamma \giv X, Z)\\
\cov(\hat\gamma,\hat\beta \giv X, Z) & \var(\hat\gamma \giv X, Z)
\end{bmatrix} \\
&\equiv \sigma^2 \begin{bmatrix}
((U^{\top}U)^{-1})_{11} & ((U^{\top}U)^{-1})_{12}\\
((U^{\top}U)^{-1})_{21} & ((U^{\top}U)^{-1})_{22},
\end{bmatrix}. \end{align*}
where we assume here that $U$ is full rank. Because $\sigma^2$ is merely a multiplicative factor independent of all relevant parameters, it is no loss of generality to take $\sigma^2 = 1$ going forward for simplicity. The treatment effect vector $\tilde{\gamma}$ is our parameter of primary interest, so we want to minimize a measure of the magnitude of
$((U^{\top}U)^{-1})_{22}$.
When $X$ is fixed, a standard choice would be to assign $\{Z_1, \ldots, Z_n\}$ to minimize the $D$-optimality criterion of 
$$\det\bigl(((U^{\top}U)^{-1})_{22}\bigr)=\prod_{j=1}^{d + 1}\lambda_j(\var(\hat\gamma \giv Z)),$$
where $\lambda_j(\cdot)$ is the $j$th eigenvalue of its argument. D-optimality is the most studied design choice among many alternatives \cite[Chapter 10]{atkinson2007optimum}.
It has the convenient property of being invariant under reparametrizations. This criterion is actually a particular case of D$_S$-optimality, in which only the parameters corresponding to a subset $S$ of the indices are of interest. See Section 10.3 of \cite{atkinson2007optimum}. Much of our theory and discussion will generalize to a broader class of criteria that we discuss in Section \ref{sssec:prospective}.

Before proceeding to the setting in which $X$ and $Z$ are random, we note a helpful result. Under the model~\eqref{eq:linmod}, there is a convenient property of D-optimality in this setting, which we state as the following simple theorem.

\begin{theorem}\label{thm:dopt=dopt}
For data following model~\eqref{eq:linmod}, assume that $U$ is full rank.
Then the D-optimality criteria for $\tilde{\gamma}$, $\tilde\beta$ and $\delta$ are equivalent.
\end{theorem}
\begin{proof}
D-optimality for $\tilde\gamma$ comes from minimizing $\det((U^{\top}U)^{-1})_{22})$ while D-optimality for $\delta$ comes from maximizing $\det(U^\top U)$. To show that those are equivalent, we write
\begin{align} \label{eq:uTu}
U^{\top}U =
\begin{bmatrix}
\tilde{X}^{\top} \tilde{X} & \tilde{X}^{\top} D \tilde{X} \\
\tilde{X}^{\top} D \tilde{X} & \tilde{X}^{\top} D^2\tilde{X} \end{bmatrix} =
\begin{bmatrix}
\tilde{X}^{\top} \tilde{X} & \tilde{X}^{\top} D \tilde{X} \\
\tilde{X}^{\top} D \tilde{X} & \tilde{X}^{\top} \tilde{X} \end{bmatrix}
\equiv
\begin{bmatrix}
A & B\\
B & A
\end{bmatrix}
\end{align}
using $D^2 = I$ in the second equality.
In the decomposition above, $A$ and $B$ are symmetric with $A$ invertible and so, from properties of block matrices
\begin{align*}
\det(U^{\top}U) &= \det(A) \det(A - BA^{-1}B) \quad \text{and}\\
((U^{\top}U)^{-1})_{22} &= (A - BA^{-1}B)^{-1},
\end{align*}
from which
\begin{equation} \label{matrixequiv} \det\left(U^{\top}U\right) = \frac{\det(\tilde{X}^{\top}\tilde{X})}{\det(((U^{\top}U)^{-1})_{22})}. \end{equation}
Because $\tilde X^{\top}\tilde X$ does not depend on Z,
our D-optimality criterion is equivalent to maximizing $\det(U^{\top}U)$. Finally, D-optimality for $\tilde\beta$ and $\tilde\gamma$ are equivalent by symmetry.
\end{proof}
The equivalence \eqref{matrixequiv} is well-known in the D$_S$-optimality literature (see, e.g., Equation 10.7 of \cite{atkinson2007optimum}). In our setting, only the right half of the columns of $U$ can be changed. Then the numerator in Equation \eqref{matrixequiv} is conveniently fixed and so optimizing the denominator alone optimizes the ratio.

The simple structure of the model~\eqref{eq:linmod} has made $D$-optimal estimation of $\tilde\gamma$ equivalent to D-optimality for $\delta$ and for $\tilde\beta$.
By the same token, a design that is A-optimal
for $\tilde\gamma$, minimizing $\tr(\var(\hat{\tilde\gamma}))$ is also A-optimal for $\delta$ and for $\tilde\beta$. Lemma 1 of \cite{nachtsheim1989design} includes our setting and shows that $D$-optimality for $\delta$ is equivalent to $D$-optimality for $\tilde\gamma$.  It does not apply to $\tilde\beta$ for our problem nor does that result consider other criteria such as A-optimality.

The theory of marginally restricted D-optimality in \cite{cookthibodeau} describes some settings where the D-optimal design for all variables is a tensor product of the given empirical design for the fixed variables and a randomized design for the settable variables. Such a design is simply an RCT on the settable variables. By their Lemma 1, this holds when the regression model is a Kronecker product of functions of fixed variables times functions of settable variables. In their Lemma 3, this holds when the regression model has an intercept plus a sum of functions of settable variables and a sum of functions of fixed variables.  Neither of those apply to model~\eqref{eq:linmod} but Lemma 2 of \cite{nachtsheim1989design} does.  The TBD designs we consider usually have constraints on the short-term gain or monotonicity constraints, and those generally make RCTs non-optimal.

\subsection{Prospective D-optimality}\label{sssec:prospective}

In this section, we modify the approach in Section \ref{sssec:fixedx} to account for the randomness in $X$ and $Z$. To do so, we adopt a prospective D-optimality criterion to apply to the setting where both $X_i$ and $Z_i$ have yet to be observed. We derive our approach from ideas in Bayesian optimal design, which we briefly summarize below based on Chapter 18.2 of \cite{atkinson2007optimum}, omitting details that will not play a role in our setting.

Bayesian optimal design often arises when the variance of the parameter estimate depends on the unknown true value of the parameter $\theta$, as is the usual case for models where the expected response is a nonlinear function of $\theta$. In this case, the information matrix $M=M(\theta)$ is usually constructed by linearizing the model form and then taking the expected outer product of the gradient with itself under a design measure on the predictors. This $M(\theta)$ is random because $\theta$ has a prior distribution. In our case, $M = U^{\top}U$ since this is the information matrix for the multivariate regression. 

The favored approach in Bayesian optimal design is to choose the design in order to minimize $\bbE( \log \det (M^{-1}))$ where the expectation is over the prior distribution on the parameters \citep{chaloner1995bayesian, dette1996}. That can be quite expensive to do.  Many of the examples in the literature optimize the design over a grid which is reasonable when the dimensions of $X$ and $\theta$ are both small, but those methods do not scale well to larger problems. 

Table 18.1 of \cite{atkinson2007optimum} lists four additional criteria along with the above one.  They are
$\log\bbE(\det(M^{-1}))$,
$\log\det(\bbE(M^{-1}))$,
$\log(\bbE(\det(M))^{-1})$, and
$\log(\det\bbE(M))^{-1}$. The objective becomes more tractable each time a  nonlinear operation is taken out of the expectation. When the logarithm is the final step, the criterion is equivalent to not taking that logarithm.

We choose the last of those four quantities (choice V in their Table 18.1) for our definition of prospective $D$-optimality.

\begin{definition}\label{prospectivedopt}[Prospective D-optimality]
For random predictors $X$, a design function $\bbP(Z=1\giv x)=p(x)$ is \textit{prospectively D-optimal} if it maximizes 
$$\det(\bbE(U^{\top}U)) = \det(\bbE[ \var(\hat\delta \giv X, Z)^{-1}]),$$
where the expectation is with respect to $X$ and $Z$.
\end{definition}

We could analogously define prospective D-optimality for $\tilde{\beta}$ or $\tilde{\gamma}$ as minimizing $\det((\bbE[U^{\top}U]^{-1})_{11})$ or $\det((\bbE[U^{\top}U]^{-1})_{22})$, respectively. By Theorem~\ref{thm:dopt=dopt}, prospective D-optimality in the sense of Definition \ref{prospectivedopt} is equivalent to these conditions in our model, so the three notions all align. 

Our choice is known as EW D-optimality in Bayesian design for generalized linear models (GLMs) \citep{YMM2016, Bu2020}. The E is for expectation and the $W$ refers to a weight matrix arising in GLMs. It is valued for its significantly reduced computational cost. While a computationally efficient design is not necessarily the best one, \cite{YMM2016} and \cite{YTM2017} find in a range of simulation studies that EW D-optimal designs tend to have strong performance under the more standard Bayesian D-optimality criterion as well. 

Since our prospective D-optimality criterion uses the expected value of $U^{\top}U$, it depends only on the $p_i$ terms and not on the joint distribution of the $Z_i$. It is therefore important that the $Z_i$ be independent in order to distinguish, say, an RCT with $\bbP(Z_i=1)=1/2$ from an allocation that takes all $Z_i=Z_1$ where $Z_1=1$ with probability $1/2$. In addition, \cite{morr:etal:2024} show in a similar problem that, when the $Z_i$ are sampled independently, the design that minimizes $\det(\bbE[M])$ also minimizes $\bbE[\det(M)]$ asymptotically as $n \to \infty$.
\cite{klugerowen} consider some stratified sampling methods that incorporate negative correlations among the $Z_i$ which makes $U^{\top}U$ come even closer to its expectation than it does under independent sampling.

Under sampling with $X_i\sim P_X$, define
\begin{align*} \tilde{\Sigma} &= \bbE\Bigl(\frac{1}{n} \tilde{X}^{\top}\tilde{X}\Bigr) \\
&= \begin{bmatrix}
1 & \bbE[X_{\bullet}]^{\top} \\
\bbE[X_{\bullet}] & \bbE[X_{\bullet}X_{\bullet}^{\top}] \end{bmatrix}, \end{align*}
where the bullet subscript denotes an arbitrary subject with $X_\bullet\sim P_X$ and $\bbP(Z_{\bullet} = 1 \st X_{\bullet}) = p(X_{\bullet})$. In addition,
\begin{align*}
\bbE\Bigl(\frac{1}{n} (\tilde{X}^{\top}D\tilde{X})_{jk}\Bigr)
&= \bbE\Bigl(\frac{1}{n} \sum_{i=1}^{n} Z_i \tilde{X}_{ij} \tilde{X}_{ik} \Bigr)\\
&= \bbE[\tilde{X}_{\bullet j} \tilde{X}_{\bullet k}(2 p(X_{\bullet}) - 1)] \end{align*}
Now let $N$ be the matrix with
\begin{align} \label{eq:generalNdef} N_{jk} &= \bbE[\tilde{X}_{\bullet j} \tilde{X}_{\bullet k}(2p(X_{\bullet}) - 1)]. \end{align}
Under our sampling assumptions
\begin{align} \label{eq:uTuapprox}
\bbE\Bigl(\frac{1}{n} U^{\top}U\Bigr) =
\begin{bmatrix}
\tilde{\Sigma} & N \\
N & \tilde{\Sigma}
\end{bmatrix}.
\end{align}
The right-hand side of~\eqref{eq:uTuapprox}  represents the expected information per observation in our tie-breaker design. Using these formulas, we obtain the following desirable result for prospective D-optimality. 

\begin{theorem}\label{thm:rctisoptimal}
Under the model~\eqref{eq:linmod} with $X_i\sim P_X$, the RCT $p(X_{\bullet}) = 1/2$ is prospectively D-optimal among all designs $p(X_{\bullet})$. Moreover, this is the unique D-optimal design of the form ~\eqref{eq:probs}.
\end{theorem}

The proofs of Theorem~\ref{thm:rctisoptimal} and of all subsequent theorems are presented in Appendix \ref{Proofs}. Theorem~\ref{thm:rctisoptimal} does not require $X_i$ to be independent though that would be the usual model. 

The theorem establishes that the RCT is  prospectively D-optimal among \textit{any} randomization scheme $\bbP(Z=1\giv X_\bullet) =p(X_\bullet)\in[0,1]$. It is not necessarily the unique optimum in this larger class.  For instance, if
$$\bbE[ X_{\bullet j}X_{\bullet k}(2p(X_\bullet)-1)]=0$$
for all $j$ and $k$, then the function $p(\cdot)$ would provide the same efficiency as an RCT since it would make the matrix $N$ in the above proof vanish.

Though we frame Theorem~\ref{thm:rctisoptimal} via prospective D-optimality, the result also holds for a broader class of criteria that we call \textit{prospectively monotone}. The basic idea of such criteria is that they only depend on the bottom-right submatrix of $\bbE[(U^{\top}U)^{-1}]$ and that they encourage this submatrix to be small in the standard ordering on positive semi-definite matrices (i.e., $\Sigma_1 \preceq \Sigma_2$ if and only if $\Sigma_2 - \Sigma_1$ is PSD). The precise definition is as follows. 

\begin{definition}\label{prospectivelymonotone}[Prospective monotonicity]
For random predictors, we say that a criterion is \textit{prospectively monotone} for $\hat{\delta}$ if it depends only on the matrix $\mathbb{E}[\Var(\hat{\tilde{\gamma}} \giv X, Z)^{-1}]$ and if the criterion increases whenever this matrix increases in the standard ordering on positive semi-definite matrices. 
\end{definition}
Because these are the only properties of prospective D-optimality we use in Theorem~\ref{thm:rctisoptimal}, the same result holds immediately for any prospectively monotone criterion. Examples include prospective A-optimality, which minimizes the quantity $\text{Tr}(\mathbb{E}[\Var(\hat{\tilde{\gamma}} \giv Z)^{-1}])$, or prospective C-optimality, which minimizes $c^{\top}\mathbb{E}[\Var(\hat{\tilde{\gamma}} \giv Z)^{-1}]c$ for some preset vector $c$ (the latter of use if a particular linear combination of elements of $\hat\gamma$ is of primary interest).  

\subsection{Symmetric Distributions with Symmetric Designs} \label{sssec:symmetric}
We can gain particular insights, both theoretical and practical, by considering a special case satisfying two conditions. First, $P_X$ has a symmetric density, i.e., $f_X(\vec{x}) = f_X(-\vec{x})$ for $\vec{x} \in \bbR^d$. This includes the special case of Gaussian covariates, which we consider in more detail as well. Secondly, we will further assume that $p = 1/2$ and the randomization window is symmetric about zero with width $\Delta\ge0$, which we call a symmetric design. That is, we restrict \eqref{eq:probs} to simply
\begin{align} \label{eq:symmetricp} 
p(Z_i = 1 \st X_i) &= \begin{cases} 1, & X_i^{\top}\eta \ge \Delta \\ 1/2, &|X_i^{\top}\eta| \in (-\Delta, \Delta) \\ 0, & X_i^{\top}\eta \leq -\Delta. \end{cases} \end{align} 
For $j=k=0$ (i.e., both terms are intercepts), equation \eqref{eq:generalNdef} reduces to
$$N_{00} = \bbE[(\indic\{X_{\bullet}^{\top} \eta \geq \Delta\} - \indic\{X_{\bullet}^{\top} \eta \leq -\Delta\})] = 0$$
since we are integrating an odd function with respect to a symmetric density. Likewise, when both $j,k\ge1$  we have $N_{jk} = 0$. The only cases that remain are the first row and first column of $N$, besides the top-left entry. Thus, we can write
\begin{align} \label{eq:Nsimple}
N = \begin{bmatrix}
0 & \alpha^{\top} \\
\alpha & \textbf{0}_{d \times d}
\end{bmatrix}
\end{align}
where $\alpha \in \mathbb{R}^d$ with
\begin{align} \label{eq:alphadef}
\alpha_j &= \bbE\bigl[X_{\bullet j} \bigl(\indic\{X_{\bullet}^{\top} \eta \geq \Delta\} - \indic\{X_{\bullet}^{\top} \eta \leq -\Delta\}\bigr)\bigr]\notag\\
&= 2\bbE[X_{\bullet j} \indic\{X_{\bullet}^{\top} \eta \geq \Delta\}].
\end{align}
We note that $\alpha = \alpha(\Delta, \eta)$ depends on the width $\Delta$ and the treatment assignment vector $\eta$, but we suppress that dependence for notational ease. From \eqref{eq:alphadef}, we can compute explicitly that
$$\label{NSiginvN} N \tilde{\Sigma}^{-1} N = \begin{bmatrix} \alpha^{\top} \Sigma^{-1} \alpha & 0 \\ 0 & \alpha \alpha^{\top} \end{bmatrix},$$
so our criterion becomes
\begin{align*}
\det(\tilde{\Sigma}) \det(\tilde{\Sigma} - N \tilde{\Sigma}^{-1}N)
&= \det(\Sigma) (1 - \alpha^{\top} \Sigma^{-1} \alpha) \det(\Sigma - \alpha \alpha^{\top}) \\
&= (1 - \alpha^{\top} \Sigma^{-1} \alpha)^2 \det(\Sigma)^2. \end{align*}
In the last line we use the formula $\det(A + cd^{\top}) = \det(A) (1 + d^{\top} A^{-1} c)$ for the determinant of a rank-one update of an invertible matrix and we also note that $\det(\tilde \Sigma)=\det(\Sigma)$.
Let $W = \Sigma^{1/2}$ so that $\var(W^{-1}x) = I$. The efficiency therefore only depends on $\alpha$ through $\alpha^{\top}\Sigma^{-1}\alpha = \Vert W^{-1}\alpha\Vert^2$.

We could also ask whether we can do better by changing our randomization scheme to allow
\begin{align} \label{eq:generalp}
\bbP(Z_i = 1 \giv X_i) =
\begin{cases}
1, & X_i^{\top}\eta \geq \Delta \\
p, & |X_i^{\top}\eta| < \Delta \\
0, & X_i^{\top}\eta \leq -\Delta
\end{cases}
\end{align}
for some other $p \neq 1/2$. While this may be a reasonable choice in practice when treatment cannot be assigned equally, it cannot provide any efficiency benefit in the symmetric case, as shown in Theorem~\ref{thm:keeppathalf} below. Just as an RCT is most efficient globally, if one is using the three level rule~\eqref{eq:generalp} then the best choice for the middle level is $1/2$ and that choice is unique under a reasonable assumption.
\begin{theorem}\label{thm:keeppathalf}
If $P_X$ is symmetric and the randomization window is symmetric around zero with width $\Delta$, then a prospectively D-optimal design of the form \eqref{eq:generalp} has $p = 1/2$. Moreover, this design is unique provided that $\mathbb{P}(|X_{\bullet}^{\top} \eta| \leq \Delta) > 0$.
\end{theorem}

An informative example is the case in which $P_X = \dnorm(0, \Sigma)$ for some covariance matrix $\Sigma$. In this case, we can compute the efficiency explicitly as a function of $\Delta$. 

\begin{theorem}\label{thm:gausmeaninfo}
Let $P_X$ be the $\dnorm(0,\Sigma)$ distribution for a positive definite matrix $\Sigma$.  For $X_i\simiid P_X$ and $Z_i$ sampled independently from~\eqref{eq:probs} for a nonzero vector $\eta\in\bbR^d$ and a threshold $\Delta\ge0$
$$
\det\bigl(\bbE\bigl( \var(\hat\delta \giv Z)\bigr)^{-1}\bigr)
=\Bigl(1 - \frac{2}{\pi} e^{\frac{-\Delta^2}{\eta^{\top}\Sigma\eta}}\Bigr)^2 \det(\Sigma)^2.
$$
\end{theorem}

From Theorem~\ref{thm:gausmeaninfo}
we find that the efficiency ratio between $\Delta = \infty$ (the RCT) and $\Delta = 0$ (the RDD) is $(1 - {2}/{\pi})^{-2} \approx 7.57$. The result in \cite{goldberger} gives a ratio of $(1-2/\pi)^{-1}$ for the variance of the slope in the case $d = 1$.  Our result is the same, though we pick up an extra factor because our determinant criterion incorporates both the intercept and the slope. Their result was for $d=1$; here we get the same efficiency ratio for all $d\ge1$.

In this multivariate setting
we see that for any fixed $\Delta>0$, the most efficient design minimizes $\eta^{\top} \Sigma \eta$, so it is the  eigenvector corresponding to the smallest eigenvalue of $\Sigma$. This represents the least ``distribution-aware'' choice, i.e., the last principal component vector, which aligns with our intuition that we gain more information by randomizing as much as possible.

\section{Short-term Gain}\label{sec:shorttermgain}

We turn now to the other arm of the tradeoff, the short-term gain. In our motivating problems, the response is defined so that larger values of $Y_i$ define better outcomes. 
Now $\bbE[Y_i]=\bbE[\tilde X_i^\top\tilde\beta]+\tilde T_i$ where $\tilde T_i = \bbE[Z_i \tilde X_i^{\top}\tilde\gamma]$. The first term in $\bbE[Y_i]$ is not affected by the treatment allocation and so any consideration of short-term gain can be expressed in terms of $\tilde T_i$. Further, $\bbE(\tilde T_i) = \bbE[Z_i\gamma_0] +T_i$ for $T_i = \bbE[Z_iX_i^\top\gamma]$. 
Now $\bbE[Z_i\gamma_0]$ only depends on our design via the expected proportion of treated subjects. This will often be fixed by a budget constraint and even when it is not fixed, it  does not depend on where specifically we assign the treatment. Therefore, for design purposes we may focus on $T_i$. 

Under the model~\eqref{eq:linmod} with treatment assignment from~\eqref{eq:generalp},
\begin{align} \label{treatmentbenefit} 
T_i = \bbE\bigl[X_i^{\top}\gamma\bigl(\indic\{X_i^{\top}\eta \geq u\} - \indic\{X_i^{\top} \eta \leq \ell\} + (2p - 1) \indic\{X_i^{\top}\eta \in (\ell, u)\}\bigr)\bigr]. 
\end{align} 
If $\eta = \gamma$, so that we assign treatment using the true treatment effect vector, then equation \eqref{treatmentbenefit} shows that the best expected gain comes from taking $u = \ell = 0$, which is an RDD centered at the mean. Moreover, the expected gain decreases monotonically as we increase $u$ beyond $0$ or decrease $\ell$ below $0$. This matches our intuition that we must sacrifice some short-term gain to improve on statistical efficiency. Ordinarily $\eta \neq \gamma$, and a poor choice of $\eta$ could break this monotonicity.

In the Gaussian case considered in Section \ref{sssec:symmetric}, we can likewise derive an explicit formula for the expected gain as a function of $u$, $\eta$, and $\gamma$. Letting $T_{\bullet} = Z_{\bullet}X_{\bullet}^{\top}\gamma$, we have
\begin{align*} \bbE[T_{\bullet}] &= \bbE[X_{\bullet}^{\top}\gamma (\indic\{X_{\bullet}^{\top}\eta \geq u\} - \indic\{X_{\bullet}^{\top}\eta \leq -u\})] \\
&= \gamma^{\top} \bbE[X_{\bullet} (\indic\{X^{\top}\eta \geq u\} - \indic\{X_{\bullet}^{\top}\eta \leq -u\})] \\
&= \gamma^{\top} \alpha.
\end{align*}
Using the formula \eqref{eq:alphagauss} for $\alpha$ in the Gaussian case, this is simply
$$\bbE[T_{\bullet}] = \sqrt{\frac{2}{\pi}} \frac{\gamma^{\top} \Sigma \eta}{\sqrt{\eta^{\top} \Sigma \eta}} \text{ } e^{\frac{-u^2}{2 \eta^{\top}\Sigma\eta}}.$$

We expect intuitively that $\eta=\gamma$
will maximize $\bbE[T_{\bullet}]$.
To verify this we start by choosing $u=u(\eta)$ in a way that keeps the proportion of data in the three treatment regions constant.
We do so by taking $u =u(\eta) =u_0\sqrt{\eta^{\top}\Sigma\eta}$ for some $u_0\ge0$, and then
$$\bbE[T_{\bullet}] = \sqrt{\frac{2}{\pi}} \frac{\gamma^{\top} \Sigma \eta}{\sqrt{\eta^{\top} \Sigma \eta}} \text{ } e^{{-u_0^2}/{2}}.$$
Let $\gamma_w=\Sigma^{1/2}\gamma$ and $\eta_w=\Sigma^{1/2}\eta$, using the same matrix square root in both cases.  Then
$$
\frac{\gamma^{\top}\Sigma\eta}{\sqrt{\eta^{\top}\Sigma\eta}}
=\frac{\gamma_w^{\top}\eta_w}{\Vert\eta_w\Vert}
$$
is maximized by taking $\eta_w=\gamma_w$ or equivalently $\eta=\gamma$. Any scaling of $\eta=\gamma$ leaves this criterion invariant.

Working under the normalization $\eta^{\top}\Sigma \eta = 1$, we can summarize our results in the Gaussian case as
\begin{align}\label{eq:deltatradeoff}
\det\Bigl( \frac{1}{n} U^{\top}U\Bigr) &= \Bigl(1 - \frac{2}{\pi} e^{-u_0^2}\Bigr)^2 \det(\Sigma)^2,\quad\text{and} \\
\label{eq:gaintradeoff} \bbE[T_{\bullet}] &= \sqrt{\frac{2}{\pi}} \gamma^{\top} \Sigma \eta \text{ }  e^{{-u_0^2}/{2}}.
\end{align}
With our normalization, $u^2 = u_0^2\eta^{\top}\Sigma\eta=u_0^2$.
Equations~\eqref{eq:deltatradeoff} and \eqref{eq:gaintradeoff} quantify the tradeoff between efficiency and short-term gain, that come from choosing $u_0$. Greater randomization through larger $u_0$ increases efficiency, and, assuming that the sign of $\eta$ is properly chosen, decreases the short-term gain.

\section{Convex Formulation}\label{sec:convex}

In this section, we return to the setting where $x_1,\ldots, x_n \in\bbR^d$ are fixed values but $Z_i$ are not yet assigned. We use lowercase $x_i$ to emphasize that they are non-random. We assume that any subset of the $x_i$ of size $d$ has full rank, as would be the case almost surely when they are drawn IID from a distribution whose covariance matrix is of full rank. The design problem is to choose $p_i=\bbP(Z_i=1)$.  

For given $x_i$, the design matrix in \eqref{eq:linmod} is
$$U =
\begin{bmatrix}
u_1(Z_1)^{\top} \\
u_2(Z_2)^{\top} \\
\vdots \\
u_n(Z_n)^{\top}
\end{bmatrix},
\quad\text{for}\quad
u_i(1) =u_{i+}\equiv\begin{bmatrix}
\tilde{x}_i\\
\tilde{x}_i
\end{bmatrix}
\quad\text{and}\quad
u_i(-1) =u_{i-}\equiv\begin{bmatrix}
\phantom{-}\tilde{x}_i\\
-\tilde{x}_i
\end{bmatrix}.
$$
Introducing $p_{i+}=p_i$ and $p_{i-}=1-p_i$ we get
\begin{align}
\bbE(U^{\top}U) &= \sum_{i = 1}^{n} \bbE(u_i(Z_i) u_i(Z_i)^{\top})
=\sum_{i = 1}^{n} (p_{i+} u_{i+} u_{i+}^{\top} + p_{i-} u_{i-} u_{i-}^{\top}).
\end{align}
Our design criterion is to choose $p_{i\pm}$ to minimize
\begin{equation*} \label{eq:objective} -\log \det(\bbE(U^{\top}U)) = -\log \det \Biggl(\sum_{i = 1}^{n} \sum_{s \in \{+, -\}} p_{is} u_{is} u_{is}^{\top} \Biggr). \end{equation*}
This problem is only well-defined for $n \ge d$, since otherwise the matrix does not have full rank for any choice of $p \in \bbR^n$ and the determinant is always zero. This criterion is convex in $\{p_{is} \giv 1 \leq i \leq n, s \in \{+, -\}\}$ by a direct match with Chapter 7.5.2 of \cite{boyd2004convex} over the convex domain with $0 \leq p_{i\pm} \leq 1$ and $p_{i+} + p_{i-} = 1$ for all $i$. 

It will be simpler for us to optimize over $q \equiv 2 p - 1 \in [-1, 1]^n$, in which case
\begin{align} \label{eq:convexopt}
-\log \det(\bbE(U^{\top}U)) = -\log \det \left(\begin{bmatrix} \sum_{i = 1}^{n} \tilde{x}_i \tilde{x}_i^{\top} & \sum_{i = 1}^{n} q_i \tilde{x}_i \tilde{x}_i^{\top} \\ \sum_{i = 1}^{n} q_i \tilde{x}_i \tilde{x}_i^{\top} & \sum_{i = 1}^{n} \tilde{x}_i \tilde{x}_i^{\top} \end{bmatrix}\right).
\end{align}

Absent any other constraints, we have seen that the RCT ($q_i = 0$, $p_i = {1}/{2}$ for all $i \leq n$) always minimizes \eqref{eq:convexopt}. The constrained optimization of \eqref{eq:convexopt} can be cast as both a semi-definite program \citep{boyd2004convex} and a mixed integer second-order cone program \citep{sagnolharman}.

This setting is close to the usual design measure relaxation. Instead of choosing $n_i=1$ point $(\tilde{x}_i,Z_i)$ for observation $i$ we make a random choice between $(\tilde{x}_i,1)$ and $(\tilde{x}_i,-1)$ for that point. The difference here is that we have the union of $n$ such tiny design problems. 

\subsection{Some useful convex constraints}

We outline a few reasonable (and convex) constraints that one could impose on this problem:

\medskip\noindent
\textbf{Budget constraint}: In practice we have a fixed budget for the treatments. For instance the number of scholarships or customer perks to give out may be fixed for economic reasons. We can impose this constraint in expectation by setting $(1/n) \sum_{i = 1}^{n} p_{i} = \mu$, where $\mu$ is some fixed average treatment rate.

\medskip\noindent
\textbf{Monotonicity constraint}: It may be reasonable to require that $p_{i}$ is nondecreasing in some running variable $r_i = f(\tilde{x}_i)$. For example, a university may require that an applicant's probability of receiving a scholarship can only stay constant or increase with their score, which is some combination of applicant variables. We can encode this as a convex constraint by first permuting the data matrix so that $r_{(1)} \leq r_{(2)} \leq \cdots \leq r_{(n)}$ and then forcing $p_{(1)} \leq p_{(2)} \leq \cdots \leq p_{(n)}$. Note that the formulation \eqref{eq:probs} satisfies this monotonicity constraint, in which case $r_i = \tilde{x}_i^{\top}\eta$.

\medskip
\noindent
\textbf{Gain constraint:} One may also want to impose that the expected gain is at least some fraction of its highest possible value, i.e.
\begin{align} \label{eq:gaincondition} \sum_{i = 1}^{n} p_i \tilde{x}_i^{\top} \eta \geq \rho \sum_{i = 1}^{n} |\tilde{x}_i^{\top}\eta|. \end{align}
The left-hand side of \eqref{eq:gaincondition} is the expected gain for this choice of $p_i$, whereas the right-hand side is the highest possible gain, which corresponds to the RDD $\bbP(Z_i = 1) = \indic\{\tilde{x}_i^{\top}\eta \geq 0\}$. Because $\gamma$ is typically not known exactly, \eqref{eq:gaincondition} computes the anticipated gain under the sampling direction $\eta$ we use. If an analyst has access to a better estimate of gain, such as a set of estimated treatment effects $\{\hat{\tau}(\tilde{x}_1),\ldots, \hat{\tau}(\tilde{x}_n)\}$ fit from prior data, they could replace \eqref{eq:gaincondition} by $\sum_{i = 1}^{n} p_i \hat{\tau}(\tilde{x}_i) \geq \rho \sum_{i = 1}^{n} |\hat{\tau}(\tilde{x}_i)|$, which is again a linear constraint on $q$. 

 \medskip\noindent
\textbf{Covariate balance constraint:}
The expected total value of the $j$th variable for treated subjects  is $\sum_i p_i \tilde{x}_{ij}$.
The constraint 
$$\sum_ip_i(\tilde{x}_{ij}-\Delta_j) \le 0$$ 
keeps the expected sample average of the $j$th variable
among treated subjects to at most $\Delta_j$ regardless
of the actual number of treated subjects. We note that sufficiently stringent covariate balance constraints may not be compatible with gain or monotonicity constraints. 

\subsection{Piecewise constant optimal designs}
Though we delay discussion of a particular example to Section \ref{sec:mimic}, we prove here an empirically-observed phenomenon regarding the monotonicity constraint. In particular, optimal solutions when a monotonicity constraint is imposed display only a few distinct levels in the assigned treatment probability. This is consistent with the one-dimensional theory of \cite{li:owen:2023}, though interestingly that paper observed the same behavior even with only a budget constraint.  It also used an approach that does not obviously extend to $d>1$.

\begin{theorem}\label{thm:fewlevels}
Consider the optimization problem
\begin{align} \label{eq:mon_constrained} \begin{split}
\underset{q}{\min} &-\log \det \left( \sum_{i = 1}^{n} \begin{bmatrix} \tilde{x}_i \tilde{x}_i^{\top} & q_i \tilde{x}_i \tilde{x}_i^{\top} \\ q_i \tilde{x}_i \tilde{x}_i^{\top} & \tilde{x}_i \tilde{x}_i^{\top} \end{bmatrix} \right) \\ 
\mathrm{s.t.}\ \  &q \in [-1, 1]^n, \\ 
&q_{\pi(1)} \leq q_{\pi(2)} \leq \cdots \leq q_{\pi(n)}, \\ 
&\overline{q} = \mu. \end{split} \end{align} 
Suppose that the non-intercept columns $x_1,\ldots, x_n$ are drawn IID from a density on $\bbR^d$. Then with probability one, the solution to \eqref{eq:mon_constrained} has at most $\binom{d + 2}{2} + 2$ distinct levels.
\end{theorem}

These upper bounds closely resemble those in \cite{cookfedorov}, who study constraints on a design of $\tilde{x}_1,\ldots,\tilde{x}_n$ themselves rather than on a treatment variable. However, theirs is an existence result for \textit{some} solution with this level of sparsity, whereas our result holds for \textit{any} solution with probability one. Moreover, their result derives from Caratheodory's theorem, whereas ours comes from a close analysis of the KKT conditions. 

The upper bounds in Theorem \ref{thm:fewlevels} are not tight. In our experience, there have typically been no more than five or six distinct levels even for $d$ as high as ten.

\section{MIMIC-IV-ED Example}\label{sec:mimic}

In this section we detail a simulation based on a real data set of emergency department (ED) patients. The MIMIC-IV-ED database \citep{mimic} provided via PhysioNet \citep{physionet} includes data on ED admissions at the Beth Israel Deaconess Medical Center between 2011 and 2019.

Emergency departments face heavy resource constraints, particularly in the limited human attention and beds available. It is thus important to ensure patients are triaged appropriately so that the patients in most urgent need of care are assigned to intensive care units (ICUs). In practice, this is often done via a scoring method such as the Emergency Severity Index (ESI), in which patients receive a score in $\{1, 2, 3, 4, 5\}$, with $1$ indicating the highest severity and $5$ indicating the lowest severity. MIMIC-IV-ED contains these values as acuity scores, along with a vector of vital signs and other relevant information about each patient.

Such a setting provides a very natural potential use case for tie-breaker designs. Patients arrive with an assortment of covariates, and hospitals acting under resource constraints must decide whether to put them in an ICU. A hospital or researcher may be interested in the treatment effect of an ICU bed; for example, a practical implication of such a question is whether to expand the ICU or allocate resources elsewhere \citep{phua2020lessmore}. It is also of interest \citep{chang2017icu} to understand which types of patient benefit more or less from ICU beds to improve this costly resource allocation. Obviously, it is both unethical and counterproductive to assign some ICU beds by an RCT.  Patients with high acuity scores must be sent to the ICU, and those with low acuity scores may be exposed unnecessarily to bad outcomes such as increased risk of acquiring a hospital-based infection \citep{kumar2018healthcare, vranas2018lowrisk}. However, it may be possible to randomize ``in the middle'', e.g., by randomizing for patients with an intermediate acuity scores such as $3$, or with scores in a range where there is uncertainty as to whether the ICU would be beneficial for that patient. Because such patients are believed to have similar severities, this would limit ethical concerns and allow for greater information gain.

\begin{figure}[t!]
    \centering
    \includegraphics[width=.9\hsize]{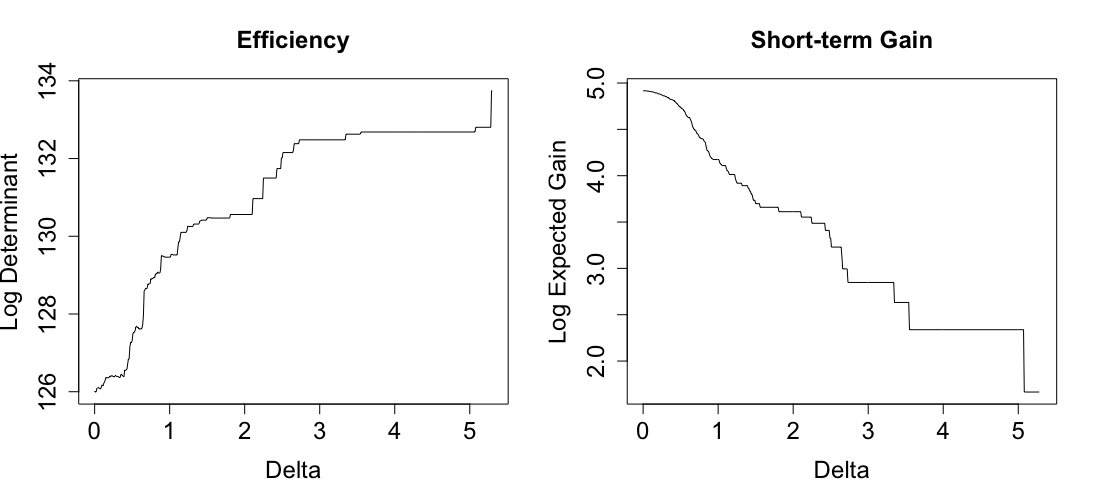}
    \caption{Efficiency and gain of the MIMIC-IV-Ed simulation as a function of the width of the randomization window $\Delta$. }
    \label{MIMICeffgain}
\end{figure}

The triage data set contains several vital signs for patients. Of these, we use all quantitative ones, which are: temperature, heart rate (HR), respiration rate (RR), oxygen saturation ($O_2$ Sat.), and systolic and diastolic blood pressure (SBP and DBP). There is also an acuity score for each patient, as described above. The data set contains 448,972 entries, but to arrive at a more realistic sample for a prospective analysis, we randomly select $200$ subjects among those with no missing or blatantly inaccurate entries. Our optimization was done using the CVXR package \citep{cvxr} and the MOSEK solver \citep{mosek}.

To carry out a full analysis of the sort described in this paper, we need a vector $\eta$, as in \eqref{eq:probs}. In practice, one could assume a model of the form \eqref{eq:linmod} and take $\eta = \hat{\gamma}$ for some estimate of $\gamma$ formed via prior data. Since we do not have any $Y$ values (which in practice could be some measure of survival, length of stay, or subsequent readmission), we will construct $\eta$ via the acuity scores, using the reasonable assumption that treatment benefit increases with more severe acuity scores.

We collapse acuity scores of $\{1, 2\}$ into a group ($Y = 1$) and acuity scores of $\{3, 4, 5\}$ into another ($Y = 0$) and perform a logistic regression using these binary groups. The covariates used are the vital signs and their squares, the latter to allow for non-monotonic effects, e.g., the acuity score might be lower for both abnormally low and abnormally high heart rates. All covariates were scaled to mean zero and variance one. For pure quadratic terms the squares of the scaled covariates were themselves scaled to have mean zero and variance one. We also considered an ordered categorical regression model but preferred the logistic regression for ease of interpretability. Our estimated $\hat{\eta_j}$ are in Table~\ref{tab:etas}.

\begin{table}
\centering
\begin{tabular}{c c c c c c c}
 \toprule
 Int. & Temp. & Temp$^2$ & HR & HR$^2$ & RR & RR$^2$ \\
  $-0.74$ & $-0.32$ & $0.22$ & $-0.03$ & $0.67$ & $-0.03$ & $\phz0.54$\\
 \midrule
 & $O_2$ Sat. & $O_2$ Sat.$^2$ & SBP & SBP$^2$ & DBP & DBP$^2$ \\
  & $\phz0.03$ & $0.36$ & $\phz0.01$ & $0.17$ & $-0.11$ & $-0.13$\\
 \bottomrule
\end{tabular}
\caption{\label{tab:etas}
These are the coefficients $\hat\eta_j$ that define a quadratic running variable for the MIMIC
data. The intercept is followed by a sum of pure quadratics in temperature, heart rate, respiration rate, $O_2$ saturation, systolic blood pressure and diastolic blood pressure.
}
\end{table}

Figure \ref{MIMICeffgain} presents the efficiency/gain tradeoff as we vary the size of the randomization window $\Delta$ in \eqref{eq:symmetricp}. For ease of visualization, we plot the logs of both quantities. As expected, we get a clear monotone increase in efficiency and decrease in gain as we increase $\Delta$, moving from an RDD to an RCT. It should be noted that our efficiency criterion, because it only uses information in the $X$ values, is robust to a poor choice of $\eta$, whereas our gain definition is constrained by the assumption that $\eta$ is a reasonably accurate stand-in for the true treatment effect $\gamma$.

In practice, it is hard to interpret what a ``good'' value of efficiency is because of our D-optimality criterion. Hence, as in \cite{owenvarian}, a pragmatic approach is to first stipulate that the gain is at least some fraction of its highest possible value, and then pick the largest $\Delta$ for this choice to maximize efficiency. A more qualitative choice based on results like Figure \ref{MIMICeffgain}, such as picking the right endpoint of a sharp efficiency jump or the left endpoint of a sharp gain decline, would also be sensible.

\begin{figure}[t!]
    \centering
    \includegraphics[width=10cm]{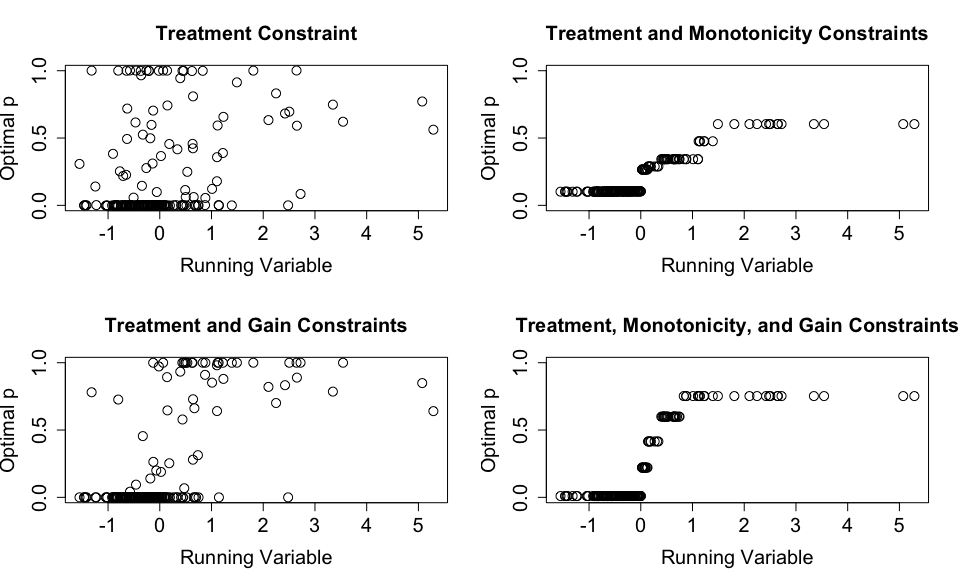}
    \caption{Optimal solutions for MIMIC-IV-ED treatment probabilities under various constraints. The treatment constraint imposed $\overline{p} = 0.2$ for the average treatment rate, and the gain constraint imposed $\rho = 0.7$, i.e., at least $70\%$ of the maximum possible gain.}
    \label{fig:MIMICplot}
\end{figure}

As we see in Figure~\ref{fig:MIMICplot}, the treatment constraint causes most $p_i$ to be at or near zero or one. Adding the gain constraint pushes most of the treatment probabilities to zero for low values of the running variable and one for high values. This scenario most closely resembles the RDD, with some deviations to boost efficiency. Indeed, the optimal solution would necessarily tend towards the RDD solution as the gain constraint increased. Finally, the monotonicity constraint further pushes the higher values of $p$ to the positive values of the running variable and vice-versa, since we lose the opportunity to counterbalance some high and low probabilities at the extreme with their opposites. The right two panels display a few discrete levels in the treatment probability, consistent with Theorem \ref{thm:fewlevels}.

\section{Discussion}\label{sec:discussion}

In this paper, we add to a growing body of work demonstrating the benefits of tie-breaker designs. Though RCTs are often infeasible, opportunities for small doses of randomization may present themselves in a wide variety of real-world settings, in which case treatment effects can be learned more efficiently. This phenomenon is analogous to similar causal inference findings about merging observational and experimental data \citep{rct+odb,rosenman2023combining,colnet2024causal}.

The convex optimization framework in Section~\ref{sec:convex} is more general and conveniently only relies on knowing sample data rather than population parameters. It is also simple to implement and allows one to incorporate natural economic and ethical constraints with ease. 

Multivariate tie-breaker designs are a natural option in situations in which there is no clear univariate running variable. For example, subjects may possess a vector of covariates, many of which could be associated with heterogeneous treatment effects in some unknown way of interest. Of course, two-line models and their multivariate analogs are not nearly as complicated as many of the models found in practice. Our view is to use them as a working model by which to decide on treatment allocations, in which case a more flexible model could be used upon full data acquisition as appropriate.

\section{Acknowledgements}
We thank John Cherian, Anav Sood, Harrison Li and Dan Kluger for helpful discussions. We also thank Balasubramanian Narasimhan for helpful input on the convex optimization problem, Michael Baiocchi and Minh Nguyen of the Stanford School of Medicine for discussions about triage to hospital intensive care units, and an anonymous reviewer for helpful comments. This work was supported by the NSF under grants IIS-1837931 and  DMS-2152780.  T.\ M.\ is supported by a B.\ C.\ and E.\ J.\ Eaves Stanford Graduate Fellowship.  

\bibliographystyle{plainnat}
\bibliography{tiebreaker}

\appendix
\section{Proofs}\label{Proofs} 

\begin{proof}[\bf Proof of \Cref{thm:rctisoptimal}]
Using $\bbE(\var(\hat\delta \giv X, Z)^{-1})$ from~\eqref{eq:uTuapprox}, we have
\begin{align*}
\det \begin{bmatrix}
\tilde{\Sigma} & N \\
N & \tilde{\Sigma}
\end{bmatrix}
&= \det(\tilde{\Sigma}) \det(\tilde{\Sigma} -  N\tilde{\Sigma}^{-1} N) \\
 &= \det(\tilde{\Sigma})^2 \det(I - \tilde{\Sigma}^{-1/2} N \tilde{\Sigma}^{-1} N \tilde{\Sigma}^{-1/2}) \\
&= \det(\tilde{\Sigma}^2) \det(I - A) \end{align*}
where $A = \tilde{\Sigma}^{-1/2} N \tilde{\Sigma}^{-1} N \tilde{\Sigma}^{-1/2}$ is symmetric and positive semi-definite. Now $\det(I - A)\le1$ with equality if and only if $A = 0$, which occurs if and only if $N = 0$. Therefore, any design giving $N = 0$ is prospectively D-optimal. From the formula \eqref{eq:generalNdef} for $N_{jk}$, we see right away that the RCT $p(X_{\bullet}) = 1/2$ satisfies this, which proves the first half of the theorem. 

If we restrict to designs of the form \eqref{eq:probs}, then for general $(\ell, u, p, \eta)$ we have
$$N_{jk} = \bbE[\tilde{X}_{\bullet j} \tilde{X}_{\bullet k}(\indic\{X_{\bullet}^{\top} \eta \geq u\} - \indic\{X_{\bullet}^{\top} \eta \leq \ell\} + q \indic\{X^{\top}\eta \in (\ell, u)\})]$$
where $q \equiv 2p - 1$. Consider a tuple $(\ell, u, p, \eta)$ for which $N = 0$. Then considering in particular the entry of $N$ where $j = k = 1$ (so that $\tilde{X}_{\bullet j} = \tilde{X}_{\bullet k} = 1$), we must have
\begin{align} \label{eq:ppart} p_u - p_{\ell} + qp_m = 0 \end{align} 
where $p_u = \bbP(X^{\top}\eta \geq u)$, $p_{\ell} = \bbP(X^{\top}\eta \leq \ell)$, and $p_m = \bbP(X^{\top}\eta \in (\ell, m))$. If we freeze $k = 0$ and vary $j$ from $1$ to $d$, we obtain
$$\bbE[X_{\bullet j} (\indic\{X_{\bullet}^{\top} \eta \geq u\} - \indic\{X_{\bullet}^{\top} \eta \leq \ell\} + q \indic\{X^{\top}\eta \in (\ell, u)\})] = 0$$
for all $j$. Taking a suitable linear combination gives
$$\bbE[X_{\bullet}^{\top}\eta (\indic\{X_{\bullet}^{\top} \eta \geq u\} - \indic\{X_{\bullet}^{\top} \eta \leq \ell\} + q \indic\{X^{\top}\eta \in (\ell, u)\})] = 0$$
Let $f(u) = \bbE[X^{\top}\eta \giv X^{\top}\eta \geq u]$, and analogously for $f(\ell)$ and $f(m)$. Then 
\begin{align} \label{eq:linearpart} f(u) p_u - f(\ell) p_{\ell} + f(m) q p_m = 0. \end{align} 
Next, equations \eqref{eq:ppart} and \eqref{eq:linearpart} give respectively that
\begin{align*} f(\ell) p_{\ell} &= f(\ell) p_u + f(\ell) q p_m \\ 
&= f(u)p_u + f(m) q p_m \end{align*} 
i.e., that 
$$(f(u) - f(\ell)) p_u = (f(\ell) - f(m)) qp_m.$$
By definition, $f(u) > f(m) > f(\ell)$, so the left-hand side is non-negative and $(f(\ell) - f(m))$ is negative. This implies that $qp_m \leq 0$. Similarly, we have
$$f(u) p_{u} = f(u)p_{\ell} - f(u) q p_m = f(\ell) p_{\ell} - f(m) q p_m$$
and so 
$$(f(u) - f(\ell)) p_{\ell} = (f(u) - f(m)) qp_m.$$
Therefore $qp_m \geq 0$,  and so $qp_m = 0$. This leaves three possibilities: $q = 0$, $p_m = 0$, or both. If $p_m = 0$, then \eqref{eq:ppart} implies that $p_u = p_{\ell} = 1/2$, which is impossible given \eqref{eq:linearpart}. So then it must be that $q = 0$, whereupon \eqref{eq:ppart} again implies that $p_u = p_{\ell}$ and \eqref{eq:linearpart} forces $p_u = p_{\ell} = 0$. In summary, a prospectively D-optimal design must satisfy
\begin{align*} &\bbP(X^{\top}\eta < \ell) = \bbP(X^{\top}\eta > u) = 0,\quad\text{and} \\
&\bbP(Z = 1 \giv X^{\top}\eta \in (\ell, u)) = \frac{1}{2}.\qedhere \end{align*} 
\end{proof}

\smallskip
\begin{proof}[\bf Proof of \Cref{thm:keeppathalf}]
Let $q = 2p - 1$. The off-diagonal block matrix $N = N(q)$ in \eqref{eq:uTuapprox} can now be written as
$$N_{jk} = \bbE[\tilde{X}_{\bullet j} \tilde{X}_{\bullet k} (\indic\{X_{\bullet}^{\top}\eta \geq u\} - \indic\{X_{\bullet}^{\top}\eta \leq u\} + q \indic\{|X_{\bullet}^{\top}\eta| < u\})].$$
That is, we can write $N = N_0 + qN_1$, where $N_0$ is as in \eqref{eq:Nsimple} and $N_1$ has $(j,k)$ entry equal to $\bbE[\tilde{X}_{\bullet j} \tilde{X}_{\bullet k} \indic\{|X_{\bullet}^{\top}\eta| < u\}]$. Note that $N_1$ is block diagonal, the exact opposite of $N_0$. Let
\begin{align} \label{objectiveq} f(q) = \log \det\bigl(\tilde{\Sigma} - (N_0 + qN_1) \tilde{\Sigma}^{-1} (N_0 + qN_1)\bigr).\end{align}
To prove the theorem, we will simply show that $f'(0) = 0$ and $f''(q) \leq 0$ for $q \in [-1, 1]$, implying that $q = 0$ (i.e., $p = 1/2$) is the global maximizer of $f$ on this interval. Let
\begin{equation} \label{ABCdef}
\begin{split}
A &= -N_1 \tilde{\Sigma}^{-1} N_1,\\ 
B &= -(N_1 \tilde{\Sigma}^{-1}N_0 + N_0 \tilde{\Sigma}^{-1}N_1), \quad\text{and}\\ 
C &= \tilde{\Sigma}^{-1} - N_0 \tilde{\Sigma}^{-1} N_0 \end{split}\end{equation}
so that $f(q) = \log \det(q^2A + qB + C)$. Call a $(d + 1) \times (d + 1)$ block matrix ``block off-diagonal'' if it is zero in the top-left entry and zero in the bottom-right $d \times d$ block, as in the case of $N_0$. The product of two block off-diagonal matrices is block-diagonal (with blocks of size $1\times 1$ and $d\times d$), and the product of a block off-diagonal matrix and such a block diagonal matrix is block off-diagonal. Thus, $A$ and $C$ are both block diagonal, whereas $B$ is block off-diagonal.
Differentiating $f$, we obtain
$$f'(q) = \tr((q^2A + qB + C)^{-1}(2qA + B))$$
so that $f'(0) = \text{tr}(C^{-1}B)$. As noted, $C$ is block diagonal and $B$ is block off-diagonal, so the product $C^{-1}B$ is block off-diagonal and thus $f'(0) = 0$.
It simplifies some expressions to let $M_1 = 2qA + B$ and $M_2 = (q^2A + qB + C)^{-1}$.
Then $f'(q) =\tr(M_2M_1)$ and
$$f''(q) = \tr(-M_1M_2M_1M_2+ 2M_2A).$$
For $q \in [-1, 1]$, $M_2$ is the upper-left block of the inverse of the covariance matrix in \eqref{eq:uTuapprox}, so it is positive semi-definite. Then $M_2^{1/2} M_1M_2 M_1M_2^{1/2}$ is positive semi-definite as well and thus
$$\tr(-M_1M_2M_1M_2) = - \tr(M_2^{1/2} M_1 M_2 M_1 M_2^{1/2}) \leq 0.$$
In addition, $A$ is negative semi-definite, so
$$\tr(2M_2 A) = 2 \tr(M_2^{1/2} A M_2^{1/2}) \leq 0.$$
Therefore, $f''(q) \leq 0$ everywhere, so $q = 0$ is in fact a global optimum.

If $\mathbb{P}(|X_{\bullet}^{\top}\eta| \leq u) > 0$, then $A$ and $B$ are both nonzero, so these two trace inequalities are strict. Then $f''(q) < 0$ for all $q \in [-1, 1]$, so $f$ cannot be constant anywhere. Since $f$ is also concave on $[-1, 1]$, the local optimum at $q = 0$ must be a global optimum on this interval.
\end{proof}

\smallskip
\begin{proof}[\bf Proof of \Cref{thm:gausmeaninfo}]
To prove Theorem \ref{thm:gausmeaninfo}, we begin with two lemmas. We write $\varphi$ for the $\dnorm(0,1)$ probability density function.  We start our study of efficiency by finding an expression for $\alpha_j$.

\begin{lemma}\label{prop:alphasphericalgaussian}
Let $P_X$ be the $\dnorm(0,I_d)$ distribution, let $\alpha_j$ be given by~\eqref{eq:alphadef} and let $Z_i$ be sampled according to~\eqref{eq:probs} for a nonzero vector $\eta\in\bbR^d$ and $u\ge0$. Then
$$
\alpha_j =2\frac{\eta_j}{\Vert\eta\Vert}\varphi\Bigl(\frac{u}{\Vert\eta\Vert}\Bigr)
$$
for $j=1,\dots,d$.
\end{lemma}
\begin{proof}
The result is easy if $\eta_j=0$.
Without loss of generality, assume that $\eta_j > 0$. Let $x_{-j}$ and $\eta_{-j}$ be the vectors in $\bbR^{d - 1}$ formed by removing the $j$th component from the vectors $x$ and $\eta$, respectively.
Using $\varphi'(t)=-t\varphi(t)$,
\begin{align*}
\bbE[x_j (\indic\{x^{\top}\eta \geq u\} )]
&= \bbE\left[x_j \indic\{x_j \geq (u - x_{-j}^{\top} \eta_{-j} )/\eta_j \} \right] \\
&= \bbE\left[ \varphi( (u-x_{-j}^{\top}\eta_{-j})/\eta_j)\right]
\end{align*}
and applying it a second time along with symmetry of $\varphi$, we get
\begin{align*}
    \alpha_j  &= \bbE[x_j (\indic\{x^{\top}\eta \geq u\} - \indic\{x^{\top}\eta \leq - u\})] \\
    &=\bbE\left[ \varphi( (u-x_{-j}^{\top}\eta_{-j})/\eta_j)
    +\varphi( (u+x_{-j}^{\top}\eta_{-j})/\eta_j)\right].
\end{align*}
Now let $\tilde{u}_j = u/\eta_j$ and $\tilde{z}_j = x_{-j}^{\top}\eta_{-j}/\eta_j \sim \dnorm( 0, \tau^2)$ with $\tau^2 = {\Vert\eta_{-j}\Vert^2}/{\eta_j^2}$.
Then we get
\begin{align*}
\alpha_j &= \frac{1}{\sqrt{2\pi}}  \frac{1}{\sqrt{2\pi \tau^2}}  \int_{-\infty}^{\infty}\left( e^{-\frac{1}{2} (\tilde{u}_j - \tilde{z}_j)^2} + e^{-\frac{1}{2} (\tilde{u}_j + \tilde{z}_j)^2}\right) e^{-\tilde{z}_j^2/2\tau^2} \rd\tilde{z}_j \\
&= \frac{1}{2\pi \sqrt{\tau^2}} \left( \frac{\sqrt{2\pi \tau^2}}{\sqrt{\tau^2 + 1}} e^{\frac{-\tilde{u}^2}{2(\tau^2 + 1)}} + \frac{\sqrt{2\pi \tau^2}}{\sqrt{\tau^2 + 1}} e^{\frac{-\tilde{u}^2}{2(\tau^2 + 1)}}\right) \\
&= \sqrt{\frac{2}{\pi}} \frac{\eta_j}{\Vert\eta\Vert_2} e^{\frac{-u^2}{2 \Vert\eta\Vert_2^2}}.\qedhere
\end{align*}
\end{proof}

\begin{lemma}\label{prop:alphageneralgaussian}
Let $P_X$ be the $\dnorm(0,\Sigma)$ distribution for a positive definite matrix $\Sigma$, let $\alpha_j$ be given by~\eqref{eq:alphadef} and let $Z_i$ be sampled according to~\eqref{eq:probs} for a nonzero vector $\eta\in\bbR^d$ and $u\ge0$. Then
\begin{align} \label{eq:alphagauss}
\alpha
= \sqrt{\frac{2}{\pi}} \frac{\Sigma \eta}{\sqrt{\eta^{\top} \Sigma \eta}} e^{\frac{-u^2}{2 \eta^{\top}\Sigma\eta}}
=2\frac{\Sigma\eta}{\sqrt{\eta^{\top}\Sigma\eta}}\varphi\biggl(\frac{u}{\sqrt{\eta^{\top}\Sigma\eta}}\biggr).
\end{align}
\end{lemma}
\begin{proof}
For the general case of $x \sim \dnorm(0, \Sigma)$, with $\Sigma$ any positive-definite matrix, we define $W = \Sigma^{1/2}$ and write $x = W z$. Then $z \sim \dnorm(0, I_d)$, and so
\begin{align*} \alpha_j &= \bbE[x_j (\indic\{x^{\top}\eta \geq u\} - \indic\{x^{\top}\eta \leq - u\})] \\
&= \bbE[W_j^{\top} (z \indic\{z^{\top} W \eta \geq u\} - \indic\{z^{\top} W \eta \leq - u\})] \\
&= W_j^{\top} \bbE[z (\indic\{z^{\top}W\eta \geq u\} - \indic\{z^{\top}W\eta \leq - u\})].
\end{align*}
This reduces the problem to the case $\Sigma = I_d$ with $\eta$ replaced by $W\eta$, so we obtain
\begin{equation*}
\alpha= W \biggl(\sqrt{\frac{2}{\pi}} \frac{W \eta}{\Vert W \eta\Vert_2} e^{\frac{-u^2}{2\Vert W \eta\Vert_2^2}} \biggr).\qedhere
\end{equation*}
\end{proof}
\noindent
We can now conclude the proof of Theorem \ref{thm:gausmeaninfo}. By Lemma~\ref{prop:alphageneralgaussian}, we have
\begin{align} 
\alpha
= \sqrt{\frac{2}{\pi}} \frac{\Sigma \eta}{\sqrt{\eta^{\top} \Sigma \eta}} e^{\frac{-u^2}{2 \eta^{\top}\Sigma\eta}}
=2\frac{\Sigma\eta}{\sqrt{\eta^{\top}\Sigma\eta}}\varphi\biggl(\frac{u}{\sqrt{\eta^{\top}\Sigma\eta}}\biggr).
\end{align}
Therefore,
\begin{align*} \det\left(\bbE\left(\frac{1}{n} U^{\top}U\right)\right) &= (1 - \alpha^{\top}\Sigma^{-1}\alpha)^2 \det(\Sigma)^2
= \left(1 - \frac{2}{\pi} e^{\frac{-u^2}{\eta^{\top}\Sigma\eta}}\right)^2 \det(\Sigma)^2. \qedhere
\end{align*}
\end{proof}

\smallskip
\begin{proof}[\bf Proof of \Cref{thm:fewlevels}]
In the optimization problem \eqref{eq:mon_constrained}, $\pi$ is a fixed permutation in $S_n$ corresponding to a monotonicity constraint, and so it is no loss of generality to take $\pi$ to be the identity and omit it. If $\mu = -1$ or $1$, then clearly the only solution is the constant vector $q = \mu$ (all treated or all control), so we ignore that case henceforth. 

Because we aim to show that optimal solutions have constant levels in $q$, it will be easier to reparametrize in terms of the increments $c_i = q_i - q_{i - 1}$ and show sparsity, where we take $c_1 = q_1$.  Letting 
$$G_0 = \begin{bmatrix} \tilde{X}^{\top}\tilde{X} & 0 \\ 0 & \tilde{X}^{\top}\tilde{X} \end{bmatrix},\quad
G_k = \sum_{i = k}^n \begin{bmatrix} 0 & \tilde{x}_i \tilde{x}_i^{\top} \\ \tilde{x}_i \tilde{x}_i^{\top} & 0 \end{bmatrix},$$ 
and $G(c) = G_0 + c_1G_1 + \cdots + c_n G_n$, we can rewrite this problem as 
\begin{align} \label{eq:mon_constrained_reparam} \begin{split}
\underset{c}{\min } &-\log \det(G(c)) \\ 
\text{such that }  &c_1 \geq -1, \\
&c_k \geq 0 \text{ for all } 2 \leq k \leq n, \\
&c^{\top}1 \leq 1, \\ 
&c_1 + \frac{n - 1}{n} c_2 + \frac{n - 2}{n} c_3 + \cdots + \frac{1}{n} c_n = \mu. \end{split} \end{align} 
The first three constraints correspond to $q \in [-1, 1]^n$ and the monotonicity condition, while the last constraint corresponds to $\overline{q} = \mu$. Let $v$ be the vector with $v_k = (n - k + 1)/n$, so that the last constraint in~\eqref{eq:mon_constrained_reparam} is $c^{\top} v = \mu$. Moreover, let 
$$A = \begin{bmatrix} -1 & 0 & \ldots & 0 \\ 0 & -1 & \ldots & 0 \\ \vdots & \vdots & \ddots & \vdots \\ 0 & 0 & \ldots & -1 \\ 1 & 1 & \ldots & 1 \end{bmatrix} \in \bbR^{(n + 1) \times n} \quad\text{and}\quad d = \begin{bmatrix} 1  \\ 0 \\ \vdots \\ 0 \\ 1 \end{bmatrix} \in \bbR^{n + 1}$$
so that the inequality constraints can be written simply as $Ac - d \leq 0$.  Slater's condition holds by considering the interior point 
$$c = \Bigl(\mu - (n - 1)\epsilon, \frac{n}{n - 1}\epsilon,\frac{n}{n-2}\epsilon,\dots,\frac{n}2\epsilon, n\epsilon\Bigr)$$ for any $\epsilon > 0$ sufficiently small, so the KKT conditions are both necessary and sufficient. The Lagrangian is 
$$L(c; \lambda, \nu) = -\log \det(G(c)) + \sum_{i = 1}^{n + 1} \lambda_i (a_i^{\top}c - d_i) + \nu (c^{\top}v - \mu),$$
where $a_i$ is a column vector formed from the $i$'th row of $A$. The differential is then 
\begin{align*} \frac{\partial L}{\partial c_k} &= -\tr(G(c)^{-1} G_k) +  \sum_{i = 1}^{n + 1} \lambda_i a_{ik} + \nu v_k \\ 
&= -2\sum_{i = k}^{n} \tilde{x}_i^{\top} G(c)^{-1}_{12} \tilde{x}_i - \lambda_k + \lambda_{n + 1} + \frac{n - k + 1}{n} \nu. 
\end{align*} 
Here, $G(c^*)^{-1}_{12} \in \bbR^{(d + 1) \times (d + 1)}$ is the off-diagonal block of $G(c^*)^{-1}$, which is symmetric since the off-diagonal block of $G(c^*)$ itself is symmetric. 

The stationarity condition is that the partial derivative above is zero for all $k$. Suppose first that, for the optimal $c^*$, \{$c_1^*, c_2^*,\ldots, c_m^*$\} are all nonzero, where $m$ is thus far unspecified. Later, we relax the condition that it is the first $m$ components of $c$ specifically that are nonzero.

The complementary slackness condition implies that $c_k^* \lambda_k = 0$ for all $1 \leq k \leq n$, so in order for this to hold we must have $\lambda_1=\lambda_2=\cdots=\lambda_m=0$. The stationarity condition for these terms is then 
\begin{equation} \label{eq:KKT} \sum_{i = k}^{n} \tilde{x}_i^{\top} G(c^*)^{-1}_{12} \tilde{x}_i = \lambda_{n + 1}^* + \frac{n - k + 1}{n} \nu^*,\quad \text{ for all } 1 \leq k \leq m.
\end{equation}
The right-hand side of \eqref{eq:KKT} is linear in $k$, with slope and intercept that depend on $\lambda_{n + 1}^*$ and $\nu^*$. These constitute two free parameters on the right-hand side. For a fixed matrix $G(c)$, this is therefore a probability zero event whenever $m \geq 3$, since $x_1,\ldots, x_n$ are sampled IID from a density. However, the issue here is that $c^*$ is not fixed but rather an implicit function of the data matrix $X$. Thus, we instead show that, for $m$ sufficiently large and with probability one, there is \textit{no} symmetric matrix $M \in \bbR^{(d + 1) \times (d + 1)}$ and pair $(\lambda_{n + 1}^*, \nu^*)$ with 
\begin{equation} \sum_{i = k}^{n} \tilde{x}_i^{\top} M \tilde{x}_i = \lambda_{n + 1}^* + \frac{n - k + 1}{n} \nu^*,\quad \text{ for all } 1 \leq k \leq m.
\end{equation}
Here, the logic is analogous: there are $\binom{d + 2}{2}$ free parameters for an arbitrary symmetric matrix $M \in \bbR^{(d + 1) \times (d + 1)}$ , as well as two additional parameters from $\lambda_{n + 1}^*$ and $\nu^*$. Thus, if $m \geq \binom{d + 2}{2} + 3$, this cannot happen with probability one. The logic is the same for any size $m$ subset of $\{\tilde{x}_1,\ldots, \tilde{x}_n\}$, not just the first $m$, so the result follows by a union bound. 

\end{proof} 

\end{document}